%% file: main.tex
\documentclass[lncs]{article}
\usepackage{times}
\usepackage{soul}
\usepackage{url}
\usepackage[hidelinks]{hyperref}
\usepackage[utf8]{inputenc}
\usepackage[small]{caption}
\usepackage{graphicx}
\usepackage{amsmath,amssymb}
\usepackage{amsthm}
\usepackage{booktabs}
\usepackage{algorithm}
\usepackage{algorithmic}
\usepackage{cleveref}

\usepackage{xparse}
\usepackage{xspace}
\usepackage{macros}
\usepackage{comment}

\newtheorem{example}{Example}
\newtheorem{theorem}{Theorem}
\newtheorem{definition}{Definition}
\newtheorem{lemma}{Lemma}

\title{On the Role of Postconditions in Dynamic First-Order Epistemic Logic}

\author{%
C\^ome Neyrand$^1$\and
Sophie Pinchinat$^2$}
\begin{document}

\maketitle

\begin{abstract} Dynamic Epistemic Logic (\DEL) is a logic that models
  information change in a multi-agent setting through the use of
  action models with pre- and post-conditions. In a recent work, \DEL
  has been extended to first-order epistemic logic (\DFOEL), with a
  proof that the resulting Epistemic Planning Problem is decidable, as
  long as action models pre- and post-conditions are non-modal and the first-order domain is finite. Our contribution
  highlights the role post-conditions have in \DFOEL. We show that the
  Epistemic Planning Problem with possibly infinite first-order
  domains is undecidable if the non-modal event post-conditions may
  contain first-order quantifiers, while, on the contrary, the problem
  becomes decidable when event post-conditions are quantifier-free. The
  latter result is non-trivial and makes an extensive use of automatic
  structures.  \end{abstract}

\section{Introduction}
%======================
\label{sec:introduction}
\input{sec-introduction}

\section{Dynamic First-Order Epistemic Logic}
%======================
\label{sec:dfoel}
\input{sec-dfoel}

\section{Epistemic Planning in \DFOEL}
%============================
\label{sec:epdfoel}
\input{sec-epdfoel}

\section{Conclusion and Future work}
%======================
\label{sec:Conclusion}

We introduced the problem Epistemic Planning Problem \EPP as a natural
problem steming from the \DFOEL setting and that allows for
infinite-domain first-order structures. Our framework thus strictly
extends the one proposed in \cite{FOEL}.

Additionally, in order to have finitely presentable inputs to \EPP, we
proposed to consider so-called \emph{automatic} epistemic models.

Regarding the problem \EPP, and inheriting from known results on
epistemic planning in the \DEL setting (no first-order logic), the
former problem, that allows for modal action models, is de facto
undecidable. However, as shown here, restricting to non-modal
action models is not sufficient in the wider first-order setting,
since the resulting subproblem \EPPnm remains undecidable (\Cref{theo:EPPun}).

\Cref{theo:EPPun} and its companion \Cref{theo:EPPzero} for the
decidability of \EPPnmqf make evidence that the nature of
post-conditions, whether quantifier-free of not, draws a line in the
decidability landscape. To our knowledge this phenomenon has never
been observed in the literature, likely because infinite first-order
domains have too eagerly been discarded.  And indeed, dealing with
infinite domain requires powerful mechanisms: our proof for
that \EPPnmqf is decidable (\Cref{theo:EPPzero}) involves a
non-trivial machinery, heavily based on automatic structures.

Even though one might wish a simpler proof for the \EPPnmqf
decidability, the chosen approach has some benefits by offering
automata constructions that entail a wide range of decidable problems
related to epistemic planning (in the same line
as \cite{DELplan}). For instance, with the proof of automaticity, we
get an automaton that represents the set of all plan solutions which
enables one to address all sorts of queries, e.g., whether they are
infinitely many, whether there exists a solution that never resorts to
some particular event, etc.

On a longer term perspective, we may investigate event schemes in
action models, as done in
\cite{FOEL}, that, for infinite first-order domains, would result in infinitely many
events, hence a powerful but challenging setting. Under some hypothesis of automaticity, we
believe that our decidability proof for \EPPnmqf can be extended.

Finally, another interesting track of research would be to allow
different domains in different words, as first-order interpretations do.
%\vfill
%\newpage
%\input{appendix}

%% The file kr.bst is a bibliography style file for BibTeX 0.99c
%\bibliographystyle{kr}
\bibliographystyle{alpha}
\bibliography{biblio}

\end{document}

%% file: sec-introduction.tex
First-order modal logic has been around since the first introduction
of modal logic. Modal logic (see \cite{blackburn}) was introduced as
an extension of propositional logic, and its semantics relies on
Kripke models \ie with different valuations in modal relations to one
another. These modal relations can be of different sorts creating a
wide range of logics that are direct applications of modal logic. For
example, temporal logic and epistemic logic are modal logics in which
modal relations are defined as relations of time and knowledge
respectively. First-order modal logic is an extension of modal logic
in which the Kripke model valuations are replaced by first-order
structures, and every modal logic can be extended to its first-order
counterpart (see \cite[Chapter 9]{blackburn}).

As a remarkable application of modal logic, Dynamic Epistemic Logic
(\DEL) \cite{van2007dynamic,sep-dynamic-epistemic}, provides a logic
that describes changes on an epistemic model -- a particular Kripke
model -- through the use of action models and updates. The links between
\DEL and Epistemic Temporal Logic (ETL) have been studied
in \cite{maubert}. Under some hypothesis on \DEL action models, the
underlying ETL model is regular, allowing to decide central problems
in artificial intelligence. A problem of particular interest is
Epistemic Planning, an automated planning problem in the DEL setting
introduced by \cite{DBLP:journals/jancl/BolanderA11}. Epistemic
Planning has been thoroughly studied (see \cite{DELplan} for a
survey). Recently, an extension of Epistemic Planning in the setting of
Dynamic First-order Epistemic Logic (\DFOEL) has been introduced
in \cite{FOEL} -- following the pioneer approach of \cite{DTML} --
where first-order logic is used as a means to compactly specify an
essentially propositional input.

In \cite{FODEL}, the authors prove that so-called non-modal Epistemic
Planning in the \DFOEL setting remains decidable provided that the first-order
domain is finite. Even though one can foresee a full propositional
encoding of the entire problem hence its decidability, the authors
of \cite{FODEL} provide an elegant proof based on
bisimulations. However, this bisimulation approach does not apply when
we relax the finiteness assumption on the first-order domain.

In this paper we investigate this Epistemic Planning Problem where
infinite domain are allowed. Discarding the hopeless case of action
models with modal (pre- and post-) conditions, we focus on pure
first-order ones. Surprisingly, we establish that even though
pre-conditions are non-modal, the decidability of the Epistemic
Planning Problem is sensitive to the nature of post-conditions. To
finitely represent our input of the problem, we require that the
first-order structures are automatic, and regarding the action model
effects, we allow arbitrary first-order interpretations, thus strictly
extending the framework of \cite{FOEL}. In this setting, we prove that
Epistemic Planning with arbitrary pure first-order post-conditions
(\EPPnm) is undecidable, while Epistemic Planning with quantifier-free
post-conditions (\EPPnmqf) is decidable. For the latter case, our
proof is involved since the infinite-domain structures generated
along histories are infinitely many, in general.

The paper is organized as follows. In Section~\ref{sec:dfoel}, we
define the full framework of \DFOEL. In Section~\ref{sec:epdfoel}, we
introduce the First-order Epistemic Planning problem arising in this framework, and we 
establish our results: the undecidability of \EPPnm in
Section~\ref{sec:EPP1} and the decidability of \EPPnmqf in
Section~\ref{sec:EPP0}. We conclude on these achievements in
Section~\ref{sec:Conclusion}.

%\spcom{a comment}

%Guess which part of the sentence Sophie has just %\spch{changed}?

%% file: sec-dfoel.tex
In this section we describe our proposal for \DFOEL, inspired
from \cite{FOEL} but that offers a wider expressiveness in both the
epistemic and the actions models: first, we allow one to consider
infinite (but still finitely presentable) first-order structures, and second, we relax the
action model post-conditions, \ie the predicate updates, as
arbitrary \emph{first-order interpretations} (in the sense of model
theory, see \cite{hodges1997shorter}) but where, for semantical reasons, the domain
remains unchanged.

We will see that this extra expressiveness yields an undecidable
epistemic planing problem. We will then better control the
expressiveness to retrieve decidability, still in a setting that
strictly extends the results from \cite{FOEL} for allowing
quantifier-free predicate updates, that in general may involve an
infinite set of tuples.

For the rest of this paper, we let $\signature$ be a first-order
signature. For a predicate $\apred \in \signature$ of arity $\arity$, we
take the convention to write $\apred(\arity) \in \signature$. Also, we let $\Var$
be a countably infinite set of \emph{variables}, whose typical elements
are $x,y,z,x_1,\ldots$. Finally, $\agt$ is a finite set of \emph{agents}.

\subsection{Preliminaries on First-Order Epistemic Logic}
%--------------------
\label{sec:pdfoel}

We restrict our definition to pure relational first-order structures
(\ie no functions in the signature)\footnote{This is no loss of expressiveness since
functions and constants can be modeled by predicates through their
graph.}.

%\begin{definition}
%  A \emph{signature} is a tuple $\sign*$ with $\Var$ a countable
%  finite set of \emph{variables}, $\Pred$ a finite set of
%  \emph{predicates} symbols.

%% and
%%   $\typeassign$ a type assignement map that
%%   satisfies:
%%   \begin{enumerate} \item For $\avar \in \Var$,
%%   $\typeassign(\avar) \in \{\agt,\obj\}$; \item For
%%   $\aconst \in \Const$, $\typeassign(\aconst) \in \{\agt,\obj\}$ \item
%%   For $\apred \in \Pred$, then $\typeassign(\apred) =
%%   (\typeassign(\apred)_1,...,\typeassign(\apred)_{ar(\apred)}) \in \{ \agt,
%%   obj \}^{ar(\apred)}$, where $\arity$ is the arity of
%%   $\apred$.  \end{enumerate} The set of terms is defined as $T :=
%%   V \cup C$.
%% \end{definition}
%\end{definition}

\begin{definition}
  The \emph{language} of First-Order Epistemic Logic (\FOEL) $\Lang$
  extends first-order logic with epistemic modalities. It is given by
  the following syntax:
  $$\Lang \ni \aformula,\aformulab ::= \atomicformula
  ~|~ \lnot \aformula ~|~ \aformula \land \aformulab ~|~ \forall
  x \aformula ~|~ \modal_\agent \aformula$$ where
  $\avar, \avar_1,...,\avar_{\arity} \in \Var$, $\agent\in \agt$ and $\apred(\arity) \in
  \signature$.  An \emph{atomic formula} is a formula of the form
  $\atomicformula$.  We denote by $\nmLang$ the pure first-order
  logic, that is the $\modal_{\agent}$-free fragment of $\Lang$, and
  by $\nmpLang$ the sub-set of $\nmLang$ with only quantifier-free
  formulas.
\end{definition}
Clearly $\nmpLang \subset \nmLang \subset \Lang$.
%\begin{example}
%\fbox{on met ?}
%\end{example}

In an epistemic logic setting, models are based on Kripke models, with
possible worlds related through epistemic relations, one for each
agent. Such a relation, say for agent $a$, specifies which worlds of
the Kripke models, agent $a$ cannot distinguish. In the richer setting
of \DFOEL, we consider \emph{first-order epistemic models} where each
possible world is assigned an entire first-order structure.
%In \DFOEL, frames are also given a domain.
\begin{definition}
\label{def:foem}
  A \emph{first-order epistemic model} (or simply \emph{epistemic
  model}) over domain $\Dom$ is a structure $\model*$ where:

  \begin{enumerate}
    \item $\Worlds$ is a non-empty set of \emph{worlds};
    \item For each $\agent \in \agt$, $\aeprel \subseteq \Worlds \times \Worlds$ is an \emph{accessibility relation}. Denote by $\aeprel(\aworld) := \{ \aworldb \in \Worlds ~|~ (\aworld, \aworldb) \in \aeprel \}$
    \item For every $\aworld \in \Worlds$, $\ainter*$ is a $\Pred$-structure associated to $\aworld$.
  \end{enumerate}
\end{definition}

Remark that in \Cref{def:foem}, the structures $\ainter$ all share the
same domain $\Dom$, but that their predicate interpretations may differ.

We provide here our running example, borrowed from formal language
 theory, that we will incrementally enrich as we progress along the
 paper.

\begin{example}
\label{ex:model}
Given regular languages $\wordlang$, $\wordlang[0]$,...,$\wordlang[n]$
 over an alphabet $\Sigma$, we first introduce the signature
 $\signatureex =
 (\predlang(1),\predlang[0](1),...,\predlang[m](1),\cdot(3),\predlangcompute(1))$,
 and the single-agent epistemic model $\modelex =
 (\{\aworldex\}, \aeprel[], \ainter)$\footnote{Notation $\aworldex$ for this single world relates to a starting computation of the language of interest.} where:
\begin{itemize}
%  \item $\Worlds := \{\aworld\}$%\{\wordlang[0],   ...,\wordlang[m] \}$

\item $\aeprel[] :=\{(\aworldex,\aworldex)\}$;%W \times W$

\item $\ainter[\aworldex]$ has domain $\Dom=\Sigma^*$ and the predicates interpretations are $\apredinterp[\predlang][\aworldex] :=
  \wordlang$, and $\apredinterp[\predlang[i]][\aworldex] := \wordlang[i]$ for
  every $0 \leq i \leq m$;

\item  $\apredinterp[\cdot][\aworldex]:= \{(\aword, \awordb, \awordc) \in (\Sigma^*)^3  ~|~ \aword = \awordb\awordc\}$;

\item $\apredinterp[\predlangcompute][\aworldex]:=\emptyset$, whose role is to interatively compute a language.
\end{itemize}
\end{example}
\begin{definition}
  Let $\model$ be an epistemic model over domain
  $\Dom$. An \emph{assignment} is a mapping $\val
  : \Var \rightarrow \Dom$, and an $\avar$\emph{-variant}
  $\val_{\avar}$ of $\val$ is an assignment \suchthat
  $\val_{\avar}(\avarb) = \val(\avarb)$ for all
  $\avarb \in \Var \sdif \{ \avar \}$.
\end{definition}

\begin{definition}
  Let $\model$ be a model and $\assignment$ be an assignment of the variables in $\Var$. The \emph{satisfaction} relation between an epistemic model $\model$ and a formula of $\Lang$ is given inductively by :
  \begin{itemize}
    \item $\mdls{\atomicformula}$ iff $(\val(\avar_1),...,\val(\avar_{\arity})) \in \apredinterp$ for all $\apred(\arity) \in \Pred$.
    \item $\mdls{\lnot \aformula}$ iff not $\mdls{\aformula}$.
    \item $\mdls{\aformula \land \aformulab}$ iff $\mdls{\aformula}$ and $\mdls{\aformulab}$.
    \item $\mdls{\forall x \aformula}$ iff $\mdls[\model][\aworld][\val_{\avar}]{\aformula}$ for every $\avar$-variant $\val_{\avar}$ of $\val$.
    \item $\mdls{\modal_{\agent} \aformula}$ iff $\mdls[\model][\aworldb][\val]{\aformula}$ for all $\aworldb \in \aeprel(\aworld)$.
  \end{itemize}
  Note that if $\aformula \in \nmLang$ \ie $\aformula$ is modal-free, then $\mdls{\aformula}$ iff $\ainter \models_{\val} \aformula$.
We write $\mdls*{\aformula}{}$ whenever for any $\val$, $\mdls{\aformula}$; in particular, if $\aformula$ has no free variables.
\end{definition}

It is clear that classical propositional epistemic logic is a fragment
of \FOEL where predicate symbols have all arity $0$.

\subsection{Dynamic First-Order Epistemic Logic}
%-----------------------
\label{sec:fdfoel}
We now enrich the setting of \FOEL with \emph{action models} that
provide the dynamics, via the \emph{update product}, similarly to the
approach in \DEL.  An action is a Kripke model whose elements
are \emph{events} with their respective \emph{precondition} and their \emph{postconditions}.

\begin{definition}
  An \emph{action model} is a tuple $\event*$
  where:
  \begin{enumerate}
\item $\Events$ is a non-empty finite set
  of \emph{events}.

\item For $\agent \in \agt$,
  $\edgeformula \subseteq \Events \times \Events$ is an accessibility
  relation.

\item $\pre : \Events \rightarrow \Lang$ assigns to each
  $\aevent \in \Events$ a \emph{precondition}, that is a
  closed\footnote{meaning without
  free-variables.} \nmLang-formula.

\item $\post
  : \Events \rightarrow (\Pred \rightarrow \Lang)$ assigns to each
  $\aevent \in \Events$ a postcondition
  $\post*(\avar_1,...,\avar_{\arity})$ for each $\apred(\arity) \in \Pred$.
\end{enumerate}
\end{definition}
We distinguish particular action models.
\begin{definition}
\label{def:nmevents}
$\event*$ is \emph{non-modal} whenever:
\begin{itemize}
\item $\pre : \Events \rightarrow \nmLang$, and
\item $\post   : \Events \rightarrow (\Pred \rightarrow \nmLang)$.
\end{itemize}
\end{definition}

\begin{example}
\label{ex:event}
We define the action model, $\eventex$ that represents the event of the application of complement, and union  and concatenation with languages among
$\wordlang[0]$, \dots, $\wordlang[m]$. Formally, we let  $\eventex:= (\Eventsex, Q, \pre, \post)$ with
\begin{itemize}
  \item $\Eventsex := \{\eventcompl\} \sunion \bigcup_{1 \leq i \leq
  m} \{ \eventunioni[i]\} \sunion \bigcup_{1 \leq i \leq
  m} \{ \eventconcati[i]\}$

\item $Q := \{(\aevent,\aevent)\,|\, \aevent \in \Events \}$;

\item for every $\aevent \in \Events$, \begin{itemize} \item
  $\pre(e) := true$;

\item for every predicate $\apred$ among
  $\predlang$, $\predlang[0]$, $\ldots,\predlang[m]$, and $\cdot$,
  $\post(\aevent)(\apred)=\apred$

\item regarding predicate $\predlangcompute$, for every $0\leq i \leq m$: \\
$$\begin{cases}
   \post(\eventunioni[i])(\predlangcompute)(x) := \predlangcompute(x) \lor \predlang[i](x) \\
    \post(\eventcompl)(\predlangcompute)(x) := \lnot \predlangcompute(x)\\
    \post(\eventconcati[i])(\predlangcompute)(x) := \exists y \exists z \predlangcompute(y) \land \predlang[i](z) \land \cdot(x,y,z)
\end{cases}$$
\end{itemize}
\end{itemize}
%\fbox{explain a bit more the role of these events}
\end{example}
The update of an epistemic state with an action model filters the
worlds that verify the preconditions and updates their structure
through the postconditions: the domain is unchanged, but predicates might be updated.
%by an arbitrary \emph{first-order interpretation}.

\begin{definition}
  Let $\model*$ and $\event*$ be given. The \emph{product update} of $\model$ and $\event$, is the epistemic model defined by
$$\upmodel*[\Worldsb][\aepfamb][\ainterfam[\aworld][\Worldsb]],$$
  where:
  \begin{itemize}
    \item $\Worldsb = \{ (\aworld,\aevent) \in \Worlds \times \Events : \mdls*{\pre*} \}$; we will simply write $\uphist[\aworld][\aevent]$ instead $(\aworld,\aevent)$.
    \item for an agent $\agent \in \agt$, $\aworld \aevent \aeprelb \aworldb \aeventb $ iff $\aworld \aeprel \aworldb$ and $\aevent \edgeformula \aeventb$.
    \item for $\aworld \aevent \in \Worldsb$, $\ainter*[\aworld \aevent][\Dom][\apred]$ where for each $\apred(\arity) \in \Pred$, $\apredinterp[\apred][\aworld \aevent]$ is the set of tuples $\tuple{\aelem}[\arity]$ such that $$\mdls[\model][\aworld][[\avar_1 \mapsto \aelem_1,...,\avar_{\arity} \mapsto \aelem_{\arity}]]{\post*(\avar_1,...,\avar_{\arity})}$$
     \end{itemize}
\end{definition}

A post-condition can be for example $\post*(\avar_1, \avar_2)
= \apred(\avar_1, \avar_2) \lor \exists \avarb
(\apred(\avar_1, \avarb) \land \apred(\avarb, \avar_2))$ which
approximates the transitive closure of the binary predicate $\apred$
as iterated updates progress.

\begin{example}
\label{ex:update}Returning to our running example on languages, the model $\modelex \otimes \eventex$ is such that:
\begin{itemize}
  \item $\Worldsb = \{\aworldex\} \times E$;

\item $\aeprelb[]
  = \{(\aworldex\aevent,\aworldex\aevent) \,|\, \aevent \in \Events\}$;

\item
  $\begin{cases}
  \apredinterp[\predlangcompute][\aworldex\eventcompl]= \Sigma^*\setminus \apredinterp[\predlangcompute][\aworldex]\\

  \apredinterp[\predlangcompute][\aworldex(\eventunioni[i])] = \apredinterp[\predlangcompute][\aworldex] \sunion \wordlang[i]\\
\apredinterp[\predlangcompute][\aworldex(\eventconcati[i])] =  \apredinterp[\predlangcompute][\aworldex] \cdot \wordlang[i]
\end{cases}
$
\end{itemize}

Thus, from world $\aworldex$, after applying the sequence of events 
$(\eventunioni[\wordlang[1]])(\eventunioni[\wordlang[2]])(\eventcompl)(\eventconcati[\wordlang[3]])$, 
one reaches a first-order structure
where $\predlangcompute$ is interpreted as 
$(\wordlang[1] \sunion \wordlang[2])^{c}\cdot \wordlang[3]$; otherwise said:
$$\apredinterp[\predlangcompute][\aworldex(\eventunioni[1])(\eventunioni[2])(\eventcompl)(\eventconcati[3])]
= (\wordlang[1] \sunion \wordlang[2])^{c}\cdot \wordlang[3]$$

One can easily enrich the setting, to capture other language operations such
as complementing inside $\wordlang[i]$, etc.
\end{example}

In the following, we aim at capturing the single infinite epistemic
model comprised of all the updates, in a way similar to the \DEL
structure introduced in \cite{DBLP:journals/corr/AucherMP14,DELplan}
for proposition \DEL.

\subsection{The infinite epistemic model of histories}
%----------------------------------------------------

Given $\model$ an epistemic model and $\event$ an action model,
  we consider the family of updates $(\itupdate)_{\aint \in \Int}$
  defined by:
  \begin{itemize}
  \item $\itupdate[0] = \model$, and
  \item
  $\itupdate[\aint+1] = \itupdate \otimes \event$.
  \end{itemize}

The epistemic $\itupdate$ is called the $\aint$-th update, and we introduce the following notations for its components:
 $$\itupdate*[\aint][\Worlds_{\aint}][\aepfam[\agent][\aint]][\ainterfam[\ahist][\Worlds_{\aint}]].$$

We call a \emph{history} any element $\ahist \in \Worlds_{\aint}$, for
some $n$, which is of the form $\ituphist[\aworld]$, where
$\aworld \in \Worlds$ and
$\aevent_1,\ldots,\aevent_\aint \in \Events$; we then say that history
$\ahist$ \emph{starts from} $\aworld$. Given $\aworld \in \Worlds$, we denote by $\Histw$ be the set of histories
starting from $\aworld$.\\

We now gather all possible iterated updates in a single infinite epistemic model.

\begin{definition}[The epistemic model of histories]
\label{def:epistemicmodelhisto}
The \emph{epistemic model of histories} is
$$\histstruct*$$
where
\begin{itemize}
\item   $\Hist  := \sUnion_{\aint \in \Int} \Worlds_{\aint}$ is
  the set of \emph{histories};

\item $\aeprel
  := \sUnion_{\aint \in \Int} \aeprel[\agent][\aint]$ is the
  accessibility relation over all the histories for agent for
  $\agent \in \agt$;

\item $\ainter[\ahist]$ is the first-order
  structure resulting from the updates induced by the sequence of
  events along history $\ahist$.
  \end{itemize}
\end{definition}

In the next section, we introduce epistemic planning problems in
 the \FOEL framework in a way that exploit the prism provided by the
 infinite model $\fohist$.

%% file: sec-epdfoel.tex
We consider a planning problem over the whole epistemic model
$\fohist$ of histories. Because the model $\model$ is infinite in
general, we rely on finite-state automata to represent it, leading to
classic \emph{automatic structures}.

More precisely, we define the class of first-order epistemic models with finitely many
worlds but whose first-order interpretations in worlds are automatic
structures, and then introduce the \emph{first-order epistemic
planning problem} in this setting.
%and discuss its decidability.

\subsection{Automatic presentations of epistemic models}
%-----------------------------------------
\label{sec:as}

Automatic structures are first-order structures with no function
symbols, such that their domain and predicates are regular, \ie
recognized by finite-state automata. An \emph{automatically presentable}
structure is a first-order structure isomorphic to an automatic
structure; the reader may refer to \cite{SURVEY} for an exhaustive survey.

\begin{definition}
An \emph{automatic presentation} over the set of predicates $\signature$ is
a finite tuple $\Apresent*$ of finite-state automata, where:
\begin{itemize}
\item $\transddom$ is an automaton over alphabet $\alphabet$, and
\item for every $\apred(\arity) \in \Pred$, $\transdi[\apred]$ is an automaton over alphabet $(\alphabet \sunion \{\pad\})^k$.
\end{itemize}
An automatic presentation $\Apresent*$ denotes a
first-order structure over signature $\Pred$ defined by $$\APstruct*$$
\end{definition}
  A first-order structure $\FOstruct*$ %over signature $\Pred$
  is \emph{automatic} if there exists a bijection $\encfunc: S \to \lang(\transddom)$, called
  an \emph{encoding},  for some
  automatic presentation $\Apresent*$, that is an isomorphism between $\FOstruct$
  and $\APstruct$.

%Automatic structures have a smooth connection with first-order logic:
\begin{theorem}[\cite{AS}]
\label{theo:FOinAS}
The first-order theory of an automatic structure is decidable.
\end{theorem}

In our setting, we allow the domain of the (first-order) epistemic
models $\model$ to be infinite as long as the first-order structure attached
to each world is automatic. Formally,

\begin{definition}[Automatic (first-order) epistemic models]
  An epistemic model $\model*$ is \emph{automatic} if $\Worlds$ is finite and $\ainter$ is automatic, for each
  $\aworld \in \Worlds$.
  \end{definition}

%% This definition is consistent with the dynamic aspects of FODEL as the
%% update of an automatically presentable epistemic model is
%% automatically presentable.

%% \begin{proposition}
%%   Let $\mathcal{M}$ be an epistemic model, say $\mathcal{M} \otimes \mathcal{E}$ is defined. \\
%%   For $w \in \mathcal{M}$, if $(D,(I(p,w))_{p \in P})$ is automatically presentable, then for $e \in \mathcal{E}$, if $(w,e) \in \mathcal{M} \otimes \mathcal{E}$, $(D,(I'(p,(w,e)))_{p \in P})$ is automatically presentable with the same encoding function.
%% \end{proposition}

%% \begin{proof}
%%   Suppose that $\mathcal{M}$ is automatically presentable with $Enc$ as encoding function. \\
%%   For $(w,e) \in \mathcal{M} \otimes \mathcal{E}$, $p \in P$, we have \\
%%   $\begin{aligned}
%%     Enc(I'(p,(w,e)) &= Enc(I(p,w) \backslash p_e^-(w) \cup p_e^+(w)) \\
%%      &= Enc(I(p,w)) \backslash Enc(p_e^-(w)) \cup Enc(p_e^+(w)) \text{ (as } Enc \text{ is injective).}
%%   \end{aligned}$ \\
%%   As $p_e^+(w)$, $p_e^-(w) \subseteq I(C,w)^{ar(\apred)}$ and $C$ is finite, we have that $Enc(p_e^+(w))$ and $Enc(p_e^-(w))$ are finite. \\
%%   As a consequence, $Enc(I(p,w))$ is regular,
%%   therefore $Enc(I'(p,(w,e))$ is regular. \\
%%   $ie.$ for $(w,e) \in \mathcal{M} \otimes \mathcal{E}$, $(D, I'(p,(w,e))_{p \in P})$ is automatically presentable with $Enc$ as the encoding function. \\
%% \end{proof}
It can be seen that the predicates updates $\post(\aevent)(\apred)$
for a non-modal action model are mere \emph{first-order
interpretations} which by \cite[Proposition 5.2]{AS} preserve
automaticity. As a consequence, we have the following.

\begin{proposition}
\label{prop:automaticupdates}
If $\model$ is an automatic and $\event$ is non-modal, then $\upmodel$ is automatic.% is also an automatic epistemic model.
\end{proposition}

\subsection{The First-order Epistemic Planning Problem}
%--------------------------
\label{sec:ep}
As originally defined in \cite{DBLP:journals/jancl/BolanderA11}, an
instance of the epistemic planning problem is composed of a
(propositional) epistemic model, a distinguished world in this model,
an action model and a first-order epistemic formula called
the \emph{goal formula}. The epistemic planning problem is to decide
whether or not there exists an executable sequence of events from the
distinguished world so that the goal formula holds.

We restate the epistemic
planning problem in order to allow inputs with possibly infinite domains.
\begin{definition}[The Epistemic Planning Problem (\EPP)] \
\begin{itemize}
\item[] {\bf Input:} an \underline{automatic} epistemic model $\model$, a distinguished world $\aworld$ in $\model$, an action model $\event$, and a first-order epistemic formula
$\goalformula$;
\item[] {\bf Output:} ``yes'' if there is a history $\ahist$ starting from $\aworld$ % \in \Histw$
such that $\fohist, \ahist \models \goalformula$, otherwise ``no''.
\end{itemize}
\end{definition}

\begin{example}
\label{ex:EPP}
We can rephrase formal language questions as epistemic planning ones,
  such as: Given a language $\wordlang[]$, can it be built using
  generators $\wordlang[0]$,...,$\wordlang[n]$ and operations among union,
  complement and concatenation? Indeed, this question amounts to
  querying the epistemic planning problem with inputs the epistemic
  model $\modelex$, the action model $\eventex$ and the epistemic
  formula $\goalformulaex$ expressing that languages
  $\predlangcompute$ and $\wordlang$
  coincide \ie $\goalformulaex := \forall x
  (\predlangcompute(x) \equivaut \predlang(x))$.
\end{example}

Epistemic planning problems have been widely investigated in the
literature \cite{DELplan,DBLP:journals/jancl/BolanderA11,MD2,DBLP:conf/ijcai/CongPS18,DBLP:journals/corr/AucherMP14,FOEL,FODEL}.
The first family of
contributions \cite{DELplan,DBLP:journals/jancl/BolanderA11,MD2,DBLP:conf/ijcai/CongPS18,DBLP:journals/corr/AucherMP14}
consider propositional epistemic and action models. This setting is
clearly captured by our framework, simply because propositional logic
can be embedded into first-order logic. It is well-known that, in the
propositional setting, the epistemic planning problem is undecidable
as soon as one allows modal event pre-condition formulas
(\cite{DELplan}). As a corollary, we can state the following.

\begin{theorem}
  $\EPP$ is undecidable.
\end{theorem}

In order to make the question about \EPP less trivial, we consider its restricted
 variant, written \EPPnm, where pre-conditions and
 post-conditions formulas are \emph{non-modal}, \ie for each event
 $\aevent$, the formulas $\pre*$ and $\post*$ are $\nmLang$-formulas.\\

%The other setting where the epistemic planning problem has been
%addressed is the one of \cite{FOEL,FODEL} with the strong restriction
%that the first-order domain of each world is finite. This assumption
%makes the epistemic planning problem of \cite{FOEL} fall into our
%problem \EPP since finite domain yield automatic epistemic
%models. In \cite{FOEL}, it is shown that under the assumption of finite
%first-order domains, the problem \EPPnm is decidable.

We below show that the problem \EPPnm is undecidable
(Theorem \ref{theo:EPPun}). This makes the decidability result
of \cite{FOEL} more anecdotal for failing in the setting of infinite
first-order domains. Additionally, we establish in \Cref{theo:EPPzero}
that there is room for a decidable subcase of \EPPnm,
written \EPPnmqf, where post-conditions of the input action model are
non-modal but also quantifier-free, \ie all formulas $\post*$ belong
to $\nmpLang$. This latter result makes the decidability of Epistemic
Planning in the (propositional action models) \DEL setting  a mere
corollary.

\subsection{\EPPnm is undecidable}%arbitrary (non-modal) quantifier predicate updates}
%------------------------------------------------------------------------------------
\label{sec:EPP1}
Because with a $\post* \in \nmLang$, we can capture the
transitive closure of a graph, we exhibit a reduction from the emptiness
problem of Turing machines (\TME). The problem \TME is known
to be undecidable \cite[Theorem 9.10]{hopcroft2001introduction} and is
defined as follows.
\begin{itemize}
\item[] {\bf Input:} a Turing machine $\TM$;% = (S, \Gamma, q_0, \delta, F)$;
\item[] {\bf Output:} do we have $\lang(\TM)=\emptyset$?
\end{itemize}

The reduction from \TME to \EPPnm  relies on the fact that the configuration
graph of a Turing machine is automatically presentable \cite[Lemma
5.1]{DBLP:journals/lmcs/KhoussainovNRS07}, and on the fact that we can
design an event that updates the binary successor predicate of this
graph in a way that this iterated updates approximates its transitive closure. \\

We now formalize this reduction.
%\begin{comment}
Let $\TM %= (Q, \Gamma, q_0, \delta, F)
$ be a Turing
machine. We consider the first-order structure $\graphTM$ over the signature comprised of
the binary predicate $\apred$, and two unary predicates $i$ and $f$
defined by: $$\graphTM :=
(\confTM, \apredinterp[\apred][\graphTM],\apredinterp[i][\graphTM],\apredinterp[f][\graphTM])$$
where its domain $\confTM$ is the set of configurations of $\TM$,
$\apredinterp[\apred][\graphTM]$ is the binary relation $\edgeTM$ composed of
configuration pairs such that $\TM$ can move from the former configuration to the
latter configuration, $\apredinterp[i][\graphTM]$ is the set of initial
configurations of $\TM$, and $\apredinterp[f][\graphTM]$ its the set of final
configurations.

%% Later, we will denote by $\edgeTM^{\leq k}$ the binary relation between
%% configurations defined by '$M$ can reach the latter from the former in
%% lesser than or equal to $k$ moves'.

As for any first-order structure, the structure $\graphTM$ can be embarked
into a single-agent single-world first-order epistemic model as follows:
$\modelTM=(\{\aworld\}, \{(w, w)\}, \graphTM)$.
%%\begin{itemize}
%  \item $\Worlds = \{ \aworld \}$
%    \item $R_a = \{ (w, w) \}$
  %% \item $\Dom=\confTM$
  %% \item $\apredinterp[i] = \{ (\alpha, 0, s_0) : \alpha \in \Gamma^* \}$
  %% \item $\apredinterp[f] = \Gamma^* \times \mathbb{N} \times F$
  %% \item $\apredinterp[p] = \,\, \edgeTM$
%%\end{itemize}

\begin{lemma}
%The first-order epistemic model
$\modelTM$ is an automatic epistemic model.
\end{lemma}

\begin{proof}
By \cite[Lemma 5.1]{DBLP:journals/lmcs/KhoussainovNRS07}, the
structure $(\confTM, \apredinterp[\apred])$ is automatically
presentable and an accurate look at the encoding of the structure used
in the proof of lemma shows that this encoding also makes the remaining relations $\apredinterp[i]$
and $\apredinterp[f]$ regular, which concludes.
\end{proof}

We next design the single-event action model
$$\eventTM = (\{\aevent\},
\{ (\aevent, \aevent) \}, \pre, \post)$$
where the effect of $\aevent$ is to augement the current interpretation
of predicate $\apred$ between configurations by its
self-composition. Formally, we let:
\begin{itemize}
%  \item $E = \{ e \}$
%  \item $\edgeformula[] = \{ (\aevent, \aevent) \}$
  \item $\pre(e) := \top$;
  \item $\post*(x_1,x_2) := \apred(x_1,x_2) \lor \exists y (\apred(x_1,y) \land \apred(y, x_2))$;
  \item $\post*[\aevent][i](x) := i(x)$ and $\post*[\aevent][f](x) := f(x)$, \ie the initial and final configurations remain unchanged.
\end{itemize}

Finally, we define the goal formula
$$\formulaTM:=\exists x \exists y (i(x) \land p(x,y) \land f(y))$$
stating the existence of an initial and a
final configuration related by $p$.

\begin{lemma}
$\TM$ is a positive instance of \TME if, and only if,
$\modelTM$, $\eventTM$, $\formulaTM$ is a positive instance
of \EPPnm.
\end{lemma}

\begin{proof}
By definition of $\post*$, a sequence $\aevent \aevent \ldots \aevent$
of updates makes predicate $p$ incrementally closer to $\edgeTM^*$,
the transitive closure of $\edgeTM$: indeed, one can show by induction
over $\ell$ that after $\ell$ triggers of event $\aevent$, two
configurations are related by (the interpretation of) $p$ if, and only
if, there is a path of length at most $\ell$ between them. Therefore,
$L(\TM) \neq\emptyset$\\
\begin{tabular}{lp{7cm}}
iff & some final configuration is
reachable from some initial configuration in $\graphTM$\\
iff & there exists a path of some length $\ell$ from some final configuration is to some initial configuration in $\graphTM$\\
iff & there exists some $\ell$ such that \\
& $\modelTM\eventTM^*, \aworld\aevent^\ell \models \exists x \exists y (i(x) \land p(x,y) \land f(y))$\\
\end{tabular} 
\end{proof}

%\end{comment}

This concludes the proof of \Cref{theo:EPPun}, which entails the following.

\begin{theorem}
\label{theo:EPPun}
The Epistemic Planning Problem \EPPnm is undecidable.
\end{theorem}

As a preliminary step towards the decidability proof of \EPPnmqf
(\Cref{theo:EPPzero}), we reduce first-order epistemic planning to an
existential first-order formula model checking query in a first-order
structure $\asStruc$ derived from the epistemic model of histories $\fohist$.

%% \begin{definition}
%%   The event model $\amod$ is \emph{non-modal} if :
%%   \begin{itemize}
%%     \item for $e \in E$, $Pre(e) \in \mathcal{L}_0$.
%%     \item for $e \in E$ and for $p(t_1,...,t_{\arity{\apred}})$ where $t_1,...,t_{\arity{\apred}} \in C$, \\
%%     $Post(e)(p(t_1,...,t_{\arity{\apred}})) \in \mathcal{L}_0$
%%     \item for $(e,e') \in E \times E$, $Q(e,e') \in \mathcal{L}_0$
%%   \end{itemize}
%% \end{definition}

%% \begin{definition}
%%   Define the following classes of epistemic planning tasks :
%%   \begin{itemize}
%%     \item $MD(n) := \{ (s_0, A, \goalformula) :$ all actions in $A$ have a modal depth lesser or equal to $n \}$
%%     \item $NM := \{ (s_0, A, \goalformula) :$ all actions in $A$ are non-modals$\}$
%%     \item $FD := \{ (s_0, A, \goalformula) : s_0$ has finite domain$\}$
%%     \item $AP := \{ (s_0, A, \goalformula) : s_0$ is automatically presentable$\}$
%%   \end{itemize}

%%   \begin{theorem}
%%     PlanEx-$(MD(2) \cap FD)$ is undecidable.
%%   \end{theorem}

%%   \begin{proof}
%%     This has been demonstrated in the case of DEL in \cite{MD2}. As DEL is a fragment of FODEL, it follows that this problem is also undecidable in FODEL. \\
%%   \end{proof}

%%   \begin{theorem}
%%     PlanEx-$(NM \cap FD)$ is decidable.
%%   \end{theorem}

%%   \begin{proof}
%%     This is the main theorem of \cite{FODEL}. \\
%%   \end{proof}
%% \end{definition}

%% %  \spch{
%%   All in all to explain that with discard modal operators in both pre and post conditions.
%%   %}

\subsection{The first-order structure of the epistemic model of histories}
%----------------------------------------------------------------------------------
\label{sec:fohist-fostruct}
\label{def:FOH}
Fix $\model*$ a first-order epistemic model, and an action model $\event*$.\\

On the basis of $\histstruct*$, we consider the domain,
called the \emph{universe}, defined by
$$\asDom = \asDom*$$ where for every history $\ahist \in \Hist$, the
set $\hDom$ is a fresh copy of domain $\Dom$ from $\model$, meant to
denote the elements of the first-order structure
$\ainter[\ahist]$.

We consider the signature
$\signAS$, obtained from the $\signature$ (the signature of $\model$)
and three other kinds of predicates, and provide their interpretations to obtain
a first-order structure $\asStruc$. Formally,
$$\signAS*$$
where: %% and introduce We list all these predicates below
%% to form signature $\signAS*$, with their interpretation in $\asStruc$.
\begin{itemize}
\item each $\epsign$ has arity $2$ and is set to:
$\apredinterp[\epsign][\asStruc]$ as the pairs of histories related by
$\aeprel$ in $\histstruct$;

\item each $\predsign$ is reshaping of predicate
 $\apred(\arity) \in \Pred$ with arity
 $\arity+1$, and set to:
 $$\apredinterp[\predsign][\asStruc]:=\{(\ahist,\aelem_1,\ldots,\aelem_{\arity}) \,|\, (\aelem_1,\ldots,\aelem_{\arity}) \in \apredinterp[\apred][\ahist]] \}$$

\item predicate $\fromw$ has arity $1$, with $\apredinterp[\fromw][\asStruc]$ equals to the subset $\Histw$ of histories that start from $\aworld$;

\item predicate $\ofDom$ has arity $2$, with $\apredinterp[\ofDom][\asStruc]$ the binary relation between any history and each element of its domain $\hDom$.
\end{itemize}

%% We can now describe the first-order structure we need to handle in our proof.
%% \begin{definition}[The structure $\asStruc$]
%% \label{def:FOH}
%% Let $\fohist$ be an epistemic model of histories, with $\Dom$ the domain shared by all the histories. For each $\ahist \in \Hist$, we let $\hDom$ be a disjoint copy of $\Dom$ and we write $\copyh : \Dom \to \hDom$ for the natural bijection, and we define the set $\asDom:=\Hist \sunion \sdUnion_{\ahist \in \Hist} \hDom$.
%% Then, the \emph{first-order structure of} $\fohist$ is the $\signAS$-structure defined by:
%% %   $$\asStruc*$$
%% \[
%% \begin{array}{rl}
%% \asHist = & (\asDom, \epinterfam*,\\
%% & \aspredinterfam*,\\
%% & \fromwinterfam*,\\
%% & \ofDominter*)
%% \end{array}
%% \]
%%   %[\asDom][\aepfam][\fami{\sUnion_{\ahist \in \Hist} (\{ \ahist \} \times \apredinterp[\apred][\ahist])}{\apred}{\Pred}][\Histwfam][\sUnion_{\ahist \in \Hist} (\{ \ahist \} \times \hDom)]$$
%% \end{definition}

The obtained structure $\asStruc$ is called the \emph{first-order
structure of} $\fohist$, and boils down to being the following
$\signAS$-structure:
\[
\begin{array}{rl}
\asHist = & (\asDom, \epinterfam*,\\
& \aspredinterfam*,\\
& \fromwinterfam*,\\
& \ofDominter*)
\end{array}
\]

The structure $\asStruc$, equipped with first-order logic over
signature $\tau = ((\epsign[\agent])_{\agent \in \agt}, \\ (\predsign)_{\apred \in \Pred}, \ofHistwfam, \ofDom)$, is at least as expressive as the epistemic model
of histories $\fohist$ (see \Cref{prop:reductionEPP-FO-auxiliary}),
via the standard translation of \FOEL into first-order logic inspired
from \cite{blackburn}, defined as follows.
\begin{definition}
\label{def:ST}
  The standard translation $\ST$ of $\Lang$ into the first-order logic over signature $\signAS$ is inductively defined by:
  \begin{itemize}
    \item $\ST(\atomicformula) :=
    \left\{\begin{array}{l}
    \atomicformulay\\
     \land \lAnd_{\aintb = 1}^{\arity} \ofDom(\avarb,\avar_{\aintb})
    \end{array}
    \right.$
    \item $\ST(\lnot \aformula) := \lnot \ST(\aformula)$
    \item $\ST(\aformula \land \aformulab) := \ST(\aformula) \land \ST(\aformulab)$
    \item $\ST(\forall \avar \aformula(\avar)) := \forall \avar (\ofDom(\avarb, \avar) \imply \ST(\aformula(\avar)))$
    \item $\ST(\modal_{\agent} \aformula) := \forall \avarb' (\epsign(\avarb, \avarb') \imply \ST[\avarb'](\aformula))$
  \end{itemize}
\end{definition}

\begin{proposition}
\label{prop:reductionEPP-FO-auxiliary}
For any assignment $\val : \Var \rightarrow \Dom$ and any history $\ahist \in \Hist$, we let $\valh : \Var \rightarrow \hDom$ be defined by $\valh(\avar):=\copyh(\val(\avar))$. Then, for any formula $\aformula \in \FOEL$,
  $$\mdls[\histstruct][\ahist][\val]{\aformula} ~\text{iff}~~ \asStruc \models_{{\valh} [\avarb \mapsto \ahist]} \ST(\aformula),$$
\end{proposition}
\begin{proof}

We proceed by induction over $\aformula$:
 \begin{itemize}
 \item $\mdls[\histstruct][\ahist][\val]{\atomicformula}$ iff $(\val(\avar_1),...,\val(\avar_{\arity})) \in \apredinterp[\apred][\ahist]$\\
 iff (by definition of $\apredinterp[\predsign][\asHist]$) $(h, \valh(\avar_1),...,\valh(\avar_{\arity})) \in \apredinterp[\predsign][\asHist]$\\
 iff
$\asHist \models_{{\valh}[\avarb \mapsto \ahist]} \predsign(y,\avar_1,...,\avar_{\arity})$\\
iff
$\asHist \models_{{\valh}[\avarb \mapsto \ahist]} \predsign(y,\avar_1,...,\avar_{\arity}) \land \lAnd_{\aintb = 1}^{\arity} \ofDom(\avarb,\avar_{\aintb})$, and this latter formula is precisely $\ST(\atomicformula)$;
 \item the cases for formulas of the form $\lnot \aformula$ and $\aformula \land \aformulab$ is smooth;
 \item $\mdls[\histstruct][\ahist][\val]{\forall x \aformula}$\\
 iff $\mdls[\model][\ahist][\val_{\avar}]{\aformula}$ for every $\avar$-variant $\val_{\avar}$ of $\val$\\
iff (by induction) $\asStruc \models_{\valh[\ahist][(\val_{\avar})] [\avarb \mapsto \ahist]}
\ST(\aformula)$ for every $\avar$-variant $\valh[\ahist][(\val_{\avar})]$ of $\valh$;
\item $\mdls[\histstruct][\ahist][\val]{\modal_{\agent} \aformula}$ iff  for all $\ahistb \in \aeprel(\ahist)$, $\mdls[\model][\ahistb][\val]{\aformula}$\\
iff   for all $\ahistb \in \aeprel(\ahist)$, $\asStruc \models_{\valh[\ahistb][\val][\avarb' \mapsto \ahistb]} \ST[\avarb'](\aformula)$ (by induction)\\
iff $\asStruc \models_{{\valh}[\avarb \mapsto \ahist]} \forall \avarb' (\epsign(\avarb, \avarb') \imply \ST[\avarb'](\aformula))$ (since $\hDom=\hDom[\ahistb]$)\\
iff $\asStruc \models_{{\valh} [\avarb \mapsto \ahist]} \ST(\modal_{\agent}\aformula)$ (by definition of $\ST$).
 \end{itemize}
\end{proof}

An immediate corollary of
Proposition~\ref{prop:reductionEPP-FO-auxiliary} is a reduction of the
entire \EPP problem into the model-checking problem over $\asStruc$
against a first-order logic:
\begin{proposition}
\label{prop:reductionEPP-FO} There exists  $\ahist$ starting from $\aworld$ such that $\histstruct, \ahist \models \goalformula$ if, and only if, $\asStruc \models \exists \avarb (\ST(\goalformula) \land \ofHistw(\avarb))$.
\end{proposition}
%% \begin{proof}
%%   \fbox{TODO}
%% \end{proof}

Notice that by, \Cref{prop:reductionEPP-FO} above and because we
proved that \EPPnm is undecidable, there must exist a structure
$\asStruc$ arising from action models with some non-quantifier-free
event post-condition that is not automatic.

We now have gathered all the material to prove our last result.

\subsection{\EPPnmqf is decidable}
%----------------------------------
\label{sec:EPP0}

We take inspiration from the methodology
of \cite{DBLP:conf/aiml/Doueneau-TabotP18,DELplan}: In order to prove
that \EPPnmqf is decidable, it is sufficient to show that the resulting 
structure $\asStruc$ is automatic (\Cref{theo:AutStrucH}). This is
because, by \Cref{prop:reductionEPP-FO}, the epistemic planning problem then
reduces to the decidable model-checking against first-order formula
over the automatic structure $\asStruc$.

The section is therefore dedicated to the proof of the
following \Cref{theo:AutStrucH}, and ends with the
statement that \EPPnmqf is decidable (\Cref{theo:EPPzero}) as a corollary.

%% In this section we show that for the restricted class of inputs allowed
%% in \EPPnmqf, the first-order structure $\asStruc$ is automatic
%% (Theorem~\ref{theo:AutStrucH}).

\begin{theorem}
  \label{theo:AutStrucH} Let $\model$ be an automatic epistemic model
  and $\event$ be an action model where all pre- and
  post-conditions are non-modal, and post-conditions are
  quantifier-free.

Then the derived first-order structure $\asStruc$  of the epistemic model of
histories (see \Cref{def:FOH}) is automatic.
\end{theorem}

%\begin{comment}
As pre- and post-conditions are non-modal, the update by an event
$\aevent$ of the predicate interpretations at some history
$\ahist \in \Hist$ only depends on $\ainter[\ahist]$ and $\aevent$. We
can thus keep track of the interpretation along $\ahistw$ after the
trigger of each event $\aevent_{\aintb}$ by remembering the current
interpretation.

We show that, in the case where all post-conditions in $\event$ are
quantifer-free, there are only finitely many different interpretations
$\ainter[\ahist]$. Otherwise said, the natural equivalence relation
$\intereq$ between histories induced ``having isomorphic
interpretation'' (\Cref{def:intereq}) yields a finite partition of
$\Hist$ (Proposition~\ref{prop:quotfinite}).

% We then define an encoding function $\encfunc$ for the entire domain
% $\asDom = \asDom*$ of the structure $\asStruc$ and show that it yields
% regular languages for $\encfunc(\asDom)$
% (Lemma~\ref{lem:Domisregular}) but also for $\encfunc(\aseprel)$,
% $\encfunc(\aspredinter)$, $\encfunc(\fromw)$ and $\encfunc(\ofDom)$
% (Lemma~\ref{lem:allareregular})

%We introduce the relation $\intereq$ that gathers the histories with the same predicate interpretation.

\begin{definition}
\label{def:intereq}
  Define the relation $\intereq \subseteq \Hist \times \Hist$ by: for all histories $\ahist\text{,}\ahistb \in \Hist$,
  $$\ahist \ \intereq \ \ahistb \text{ iff
  } \apredinterp[\apred][\ahist]
  = \apredinterp[\apred][\ahistb] \text{, for every
  } \apred \in \Pred.$$ We will denote by $\intereqcl = \{ \ahistb
  ~|~ \ahist \ \intereq \ \ahistb \}$ the $\intereq$-equivalence class
  of history $\ahist$ and by $\quotcl$ the set of all the
  equivalence classes, with typical element $\aclass$.
\end{definition}

By \Cref{def:intereq} of $\intereq$, it is clear that
  $\ainter[\ahistb] = \ainter[\ahist]$\footnote{up to the isomorphism
  arising from the natural one-to-one mapping between $\Dom[\ahist]$
  and $\Dom[\ahistb]$.}, for every $\ahistb \in \intereqcl$. This
  allows us to consistently set for a $\intereq$-class $\aclass \in \quotcl$,
$\intereqcl[\aclass] := \ainter[\ahist]$, for some $\ahist \in \aclass$, and to define $\ainter[\aclass] := \ainter[\ahist]$.
%  \ie $\apredinterp[\apred][\intereqcl]
%  := \apredinterp[\apred][\ahist]$, for each $\apred \in \Pred$.\\
%
%As a consequence, a $\intereq$-class of histories $\intereqcl$ is fully
%determined by the first-order structure $\ainter[\ahist]$.\\

%\begin{definition}
We introduce some technical material that will be usefull in the proof that the set $\intereqcl$ is regular (\Cref{prop:classisregular}).
  Let $\ainter[]$ be a $\Pred$-interpretation with domain $\Dom$, and
  $\aevent \in \Events$, if $\ainter[] \models \pre*$, we define
  the $\Pred$-interpretation
  $$\ainter[] \otimes \aevent := (\Dom, (\apredinterup)_{\apred \in \Pred})$$
  where $\apredinterup$ is the set of $\tuple{\aelem}[\arity]$ such that:
  $$\ainter[] \mdlsz \post*\tuple{\avar}[\arity] \}.$$
%\end{definition}

\begin{lemma}
  \label{lem:concatop}
  For every $\ahist \aevent \in \Hist$, we have:
  $$\intereqcl[\ahist \aevent] = \intereqcl[\ahist] \otimes \aevent.$$
\end{lemma}
\begin{proof}
  As $\post*$ is non-modal,
  \begin{center}
   $\ainter[\ahist\aevent] \mdlsz \apred\tuple{\avar}[\arity]$\\
  if, and only if,\\
  $\fohist, \ahist \aevent \mdlsz \apred\tuple{\avar}[\arity]$\\
  if, and only if,\\
  $\fohist, \ahist \mdlsz \post*\tuple{\avar}[\arity]$\\
  if, and only if,\\
  $\ainter[\ahist] \mdlsz \post*\tuple{\avar}[\arity]$\\
  if, and only if,\\
  $\ainter[\ahist] \mdlsz \apredinterup\tuple{\avar}[\arity]$
  \end{center}
\end{proof}

%% Finiteness of $\quotcl$ is a
%% milestone in showing the automaticity of $\asHist$, as it allows us to
%% rely on finitely many different interpretations along all the possible
%% histories.

Recall that in this section, $\model$ is automatic and $\event$ is an
  action model where all pre- and post-conditions are non-modal, and
  post-conditions are quantifier-free.

We use this assumption to
  establish the following, essentially stating that the set
  $\ainterfam[\ahist][\Hist]$ is a finite family.

\begin{proposition}
\label{prop:quotfinite}
  $\quotcl$ is finite.
\end{proposition}

In order to establish \Cref{prop:quotfinite}, it is sufficient to show that every intepretations $\apredinterp[\apred][\ahist]$
belongs to the Boolean algebra finitely generated by the
interpretations $\apredinterp[\apredb]$, where $\apredb \in \Pred$ and
$\aworld \in \Worlds$. Note that since this algebra of sets is
finitely generated, it is itself finite.

We first introduce some notations.
\begin{definition} \label{def:atom} 
\begin{itemize}
\item For $\aint \in \Int$, let $\intset := \{ 1,...,\aint \}$.

\item For $\arityb, \arity \in \Int$, $\aset \subseteq {\Dom}^{\arityb}$, and 
    $\mix: \intfunc[\arityb][\arity]$, we let:
$$\aset \mix
    := \{ \tuple{\aelem}[\arity] ~|~
    (\aelem_{\mix(1)},...,\aelem_{\mix(\arityb)}) \in \aset \}$$

\item For $G=\{\aset[1],\ldots,\aset[m]\}$ comprised of subsets of a given superset, denote by
    $\bool G$ the boolean algebra generated by $G$.
    \end{itemize}
        \end{definition}
We now show that the interpretation of a
predicate after an update is a combination of atoms from the precedent
interpretations.
  \begin{lemma}
    \label{lem:IH}
Let $\ahist \aevent \in \Hist$,  and $\apred(\arity) \in \Pred$.

If $\post* \in \nmpLang$, then
    $$\apredinterp[\apred][\ahist \aevent] \in \bool \{ \apredinterp[\apredb][\ahist] \mix ~|~ \apredb(\arityb) \in \Pred, ~\sigma : \intfunc \}.$$
  \end{lemma}
  \begin{proof} We reason by induction over formula $\post* \in \nmpLang$.
    If $\post*$ is an atomic formula of $\nmpLang$, say of the form $\apredb(\avar_{\mix(1)},...,\avar_{\mix(\arityb)})$, we have: \\
    $\begin{array}{l}
      \tuple{\aelem}[\arity] \in \apredinterp[\apred][\ahist \aevent]\\
      \text{iff}~ \ainter[\ahist] \mdlsz[\arity] \apredb(\avar_{\mix(1)},...,\avar_{\mix(\arityb)}) \\
      \text{iff}~ (\aelem_{\mix(1)},...,\aelem_{\mix(\arityb)}) \in \apredinterp[\apredb][\ahist] \\
      \text{iff}~ \tuple{\aelem}[\arity] \in \apredinterp[\apredb][\ahist] \mix.
      \end{array}
    $ \\
  Otherwise $\post*$ is Boolean combination of atomic formulas of $\nmpLang$, and as a result:
  $$\apredinterp[\apred][\ahist \aevent] \in \bool \{ \apredinterp[\apredb][\ahist] \mix ~|~ \apredb(\arityb) \in \Pred, ~\sigma : \intfunc \}$$
  \end{proof}

We now establish that the interpretation of each predicate
  $\apredinterp[\apred][\ahist]$ for a history $\ahist$ starting from
  $\aworld$ belongs to the Boolean algebra generated by the initial
  interpretations $\apredinterp[\apredb]$ ($\apredb \in \signature$) of
  the predicates in $\ainter$.

\begin{lemma} \label{lem:inBool} For every $\ahist \in \Hist$, for every
  $\apred \in \Pred$,
    $$\apredinterp[\apred][\ahist] \in \bool \{ \apredinterp[\apredb] \mix ~|~ \apredb(\arityb) \in \Pred, ~\sigma : \intfunc \}.$$
  \end{lemma}
  \begin{proof}
First, observe that
  %% \begin{lemma}
  %%   \label{lem:compo}
given $\arity, \arityb, \arityc \in \Int$, $\mix : \intfunc[\arityb][\arity]$, $\mixb : \intfunc[\arityc][\arityb]$, and $\aset \subseteq {\Dom}^{\arityc}$, we have:
    \begin{equation}
    \label{eq:compo}
    (\aset \mixb)\mix = \aset (\mix \compo \mixb)
    \end{equation}
%%  \end{lemma}
%%  \begin{proof}
Indeed, $\tuple{\aelem}[\arity] \in (\aset \mixb) \mix$ iff $(\aelem_{\mix(1)},...,\aelem_{\mix(\arityb)}) \in \aset \mixb$ iff $(\aelem_{\mix(\mixb(1))},...,\aelem_{\mix(\mixb(\arityc))}) \in \aset $ iff $\tuple{\aelem}[\arity] \in \aset (\mix \compo \mixb)$, which achieves the argument. \\
%%  \end{proof}

We prove \Cref{lem:inBool} by induction on $\ahist$.

If $\ahist=\aworld$, then for every $\apred(\arity) \in \Pred$, we clearly have $\apredinterp
    \in \bool \{ \apredinterp[\apredb] \mix ~|~ \apredb(\arityb) \in \Pred, ~\sigma : \intfunc \}$.

Otherwise, for a history of the form  $\ahist \aevent$, by \Cref{lem:IH}, we have that
    $\apredinterp[\apred][\ahist \aevent] \in \bool \{ \apredinterp[\apredb][\ahist] \mix
    ~|~ \\ \apredb(\arityb) \in \Pred, ~\sigma : \intfunc \}$. By induction hypothesis,
    $\apredinterp[\apredb][\ahist] \in \bool \{ \apredinterp[\apredc] \mixb
    ~|~ \apredc(\arityc) \in \Pred, ~\mixb
    : \intset[{\arityc}] \to \intset[{\arityb}] \}$.
    Thus $\apredinterp[\apredb][\ahist] \mix \in \bool \{
    (\apredinterp[\apredc] \mixb) \mix ~|~ \apredc \in \Pred, ~\mixb
    : \intset[{\arityc}] \to \intset[{\arityb}] \}$, Therefore by \Cref{eq:compo},
    $$\apredinterp[\apredb][\ahist] \mix \in \bool \{ \apredinterp[\apredc](\mix \compo \mixb) ~|~ \apredc(\arityc) \in \Pred, ~\mixb
    : \intset[{\arityc}] \to \intset[{\arityb}] \},$$ \ie
    $\apredinterp[\apredb][\ahist] \mix \in \bool \{ \apredinterp[\apredc]\mixb
    ~|~ \apredc(\arityc) \in \Pred, ~\mixb
    : \intset[{\arityc}] \to \intset[{\arity}] \}$, which entails
    $$\apredinterp[\apred][\ahist \aevent] \in \bool \{ \apredinterp[\apredb]\mix
    ~|~ \apredb(\arityb) \in \Pred, ~\mix
    : \intset[{\arityb}] \to \intset[{\arity}] \}.$$ \end{proof}

Since $\{ \apredinterp[\apredb]\mix ~|~ \apredb[\arityb] \in \Pred, ~\mix : \intset[{\arityb}] \to \intset[{\arity}] \}$ is finite, so is $\bool \{ \apredinterp[\apredb]\mix ~|~ \apredb(\arityb) \in \Pred, ~\mix : \intset[{\arityb}] \to \intset[{\arity}] \}$. Therefore there is a finite number of different $\apredinterp[\apred][\ahist]$, which achieves the proof of \Cref{prop:quotfinite}.

This finiteness property paves the way to defining an automaton that
while reading a history $\ahist$ can determine $\ainter[\ahist]$
(Proposition~\ref{prop:classisregular}).  %% Still, knowing that each
%% interpretation $\ainter[\ahist]$, or equivalently $\ainter[\ahist]$ is
%% automatic (Proposition~\ref{prop:automaticupdates}) does not provide
%% us with the mechanism to know which interpretation to consider after
%% history $\ahist$. The following proposition is an answer.

\begin{proposition}
  \label{prop:classisregular}
  For each $\ahist \in \Hist$, $\intereqcl$ is regular.
\end{proposition}

\begin{proof}
%Before proving Proposition~\ref{prop:classisregular}, we introduce some material.
%As each class $\intereqcl$ is associated to an interpretation, we can write $\intereqcl \otimes \aevent$ for $\aintercl \otimes \aevent$.
We establish the regularity of $\intereqcl$ by constructing a
finite-state automaton $\automata[\intereqcl]$ for it.

This automaton
(over alphabet $\Worlds \sunion \Events$) has states ranging over
$\{ \initstate \} \sunion \quotcl$ (where $\initstate$ is a fresh
state), initial state $\{ \initstate \}$, and final states ranging over
$\{ \intereqcl \}$, and the transition relation $\delta$
defined by:
\begin{itemize}
\item Regarding $\initstate$:
\begin{itemize}
\item $\transition(\initstate, \aworld) = \intereqcl[\aworld]$, for every $\aworld \in \Worlds$;
\item $\transition(\initstate, \aevent) = \emptyset$, for $\aevent \in \Events$;
\end{itemize}
\item Regarding every $\aclass \in \quotcl$:
\begin{itemize}
\item  $\transition(\aclass, \aworld) = \emptyset$, for every $\aworld \in \Worlds$;
\item  $\transition(\aclass, \aevent) =
\left\{
\begin{array}{ll}
\aclass \otimes \aevent & \text{if }\ainter[\aclass] \models \pre*,\\
\emptyset & \text{otherwise.}
\end{array}
\right.$
\end{itemize}
\end{itemize}
%% We define automaton $\automata*[\intereqcl][\Worlds \sunion \Events]$ over finite alphabet with: \fbox{paragraph}

%%   \begin{itemize}

%%     \item $\states = \{ \initstate \} \sunion \quotcl$
%%     \item For $\aworld \in \Worlds$, $\transition(\initstate, \aworld) = \intereqcl[\aworld]$ \\
%%     For $\aclass \in \quotcl$ and $\aevent \in \Events$, if $\ainter[\aclass] \models \pre*$, then $\transition(\aclass, \aevent) = \ainter[\aclass] \otimes \aevent$
%%     \item $\initial = \{ \initstate \}$
%%     \item $\final = \{ \intereqcl \}$
%%   \end{itemize}
\noindent We now show that $\intereqcl=\lang(\automata[\intereqcl])$:\\
%% =======
%%   As $\post*$ is non-modal, we have that: \\
%%   $$\Hist, \ahist \mdlsz \post*\tuple{\avarc}[\arity{\apred}] ~\text{iff}~ \ainter[\ahist] \mdlsz \post*\tuple{\avarc}[\arity{\apred}]$$
%%   Therefore $\apred^{\ainter[\ahist] \otimes \aevent} = \apredinterp[\apred][\ahist \aevent]$ \ie $\ainter[\ahist \aevent] = \aintercl[\ahist] \otimes \aevent$.
%% \end{proof}

%% %As each class $\intereqcl$ is associated to an interpretation, we can write $\intereqcl \otimes \aevent$ for $\aintercl \otimes \aevent$.
%% \begin{lemma}
%%   \label{lem:classisregular}
%%   For each $\ahist \in \Hist$, $\intereqcl$ is regular.
%% \end{lemma}
%% \begin{proof}
%%   Define $\automata[\intereqcl]$ over finite alphabet $\Worlds \sunion \Events$, with a finite set of states $\{ \initstate \} \sunion \{ \ainter[\aclass] ~|~ \aclass \in \quotcl \}$, with an initial state $\initstate$, a final state $\ainter[\intereqcl]$ and a transition function $\transition$ defined by:
%%   \begin{itemize}
%%     \item For $\aworld \in \Worlds$, $\transition(\initstate, \aworld) = \ainter[\intereqcl[\aworld]]$
%%     \item For $\aclass \in \quotcl$ and $\aevent \in \Events$, if $\ainter[\aclass] \models \pre*$, then $\transition(\ainter[\aclass], \aevent) = \ainter[\aclass] \otimes \aevent$
%%   \end{itemize}
%%   We argue that $\lang(\automata[\intereqcl]) = \intereqcl$: \\
%% >>>>>>> f03c06f9536286bb569a1059aa0c3b2d711c41ed
$\ahistw \in \intereqcl$

$\begin{array}{lll}

    \text{iff}~ &\intereqcl[\ahistw]= \intereqcl \\
    \text{iff}~ &\intereqcl[\ahistw]\text{ is final} \\
    \text{iff}~ &\intereqcl[\aworld] \otimes \aevent_1 \otimes ... \otimes \aevent_{\aint} ~\text{is final} & \text{(\Cref{lem:concatop})}\\
    \text{iff}~ &\transition^*(\initstate,\ahistw) ~\text{is final} \\
    \text{iff}~ &\ahistw \in \lang(\automata[\intereqcl])
    \end{array}$
\end{proof}

We have here gathered all the material to show that the structure $\asStruc$
%%$$\asStruc*[\asDom][\asepfam][\aspredinterfam][(\fromwinter)_{\aworld \in \Worlds}][\ofDominter].$$
is automatic.

We first start with the encoding $\encfunc$ (see
(\Cref{def:encoding}) of the elements of the universe $\asDom=\asDom*$, then we prove that encoding $\encfunc$ provides an
automatic presentation of the first-order structure $\asStruc$ (\Cref{lem:allareregular}).

Recall that because event updates are particular cases of first-order
interpretations, and, according to
\Cref{prop:automaticupdates}, we know that each $\ainter[\ahist]$ is
automatic. As a consequence, each $\ainter[\intereqcl]$ is automatic,
say with encoding $\encfunccl$ over some
alphabet $\alphabet_{[\ahist]}$.

%% Therefore,
%% there exists an automatic presentation for $\ainter[\intereqcl]$, over some
%% alphabet $\alphabet_{[\ahist]}$, and we denote by $\encfunccl$ the
%% encoding function of this automatic presentation which maps every
%% element of $\Dom[{\ainter[\intereqcl]}]$ (that is $\Dom$) onto a finite word of
%% $\alphabet_{[\ahist]}^*$.

Now, the overall encoding function $\encfunc$ of the universe $\asDom$ is as follows.

\begin{definition}
\label{def:encoding}
 The encoding function $\encfunc$ of the universe $\asDom=\asDom*$
$$\encfunc : \asDom \to (\Worlds \sunion \Events \sunion \sUnion_{[\ahist] \in \quotcl} \alphabet_{[\ahist]})^*$$
  is defined by:
  \begin{itemize}
    \item for $\ahistw[m] \in \Hist$, $\encfunc(\ahistw[m]) := \ahistw[m]$
    \item for $\aelem \in \hDom$, $\encfunc(\aelem) := \ahist \cdot \encfunccl(\decopyh(\aelem))$; we recall that $\copyh$ is the bijection between $\Dom$ and $\hDom$.
  \end{itemize}
\end{definition}

 We recall that
$$\asStruc*[\asDom][\asepfam][\aspredinterfam][(\fromwinter)_{\aworld \in \Worlds}][\ofDominter]$$
where
\begin{itemize}
\item $\asDom = \asDom*$;
\item$\aseprel = \aeprel$, for $\agent \in \agt$;
\item $\aspredinter = \aspredinter*$, for $\apred \in \Pred$;
\item $\fromwinter = \Histw$, for $\aworld \in \Worlds$; and
\item $\ofDominter=\ofDominter*$.
\end{itemize}

\begin{lemma}
  \label{lem:Domisregular}
  $\encfunc$ is injective and $\encfunc(\asDom)$ is regular.
\end{lemma}
\begin{proof}
Mapping $\encfunc$ is injective since any element of
$\asStruc$ is the element of domain $\hDom$, for some $\ahist$, and
is therefore univoquely identified by the prefix $\ahist$ in its
encoding.

Regarding the regularity of $\encfunc(\asDom)$, we have:

$\begin{aligned} \encfunc(\asDom)
  &= \encfunc(\asDom*)
  = \encfunc(\Hist) \sunion \sdUnion_{\ahist \in \Hist} \encfunc(\hDom) \\
  &= \Hist \sunion \sdUnion_{\ahist \in \Hist} \ahist \cdot \encfunccl(\decopyh(\hDom))\\
  & (\text{by definition of } \encfunc)\\
  &= \Hist \sunion \sdUnion_{\ahist \in \Hist} \ahist \cdot \encfunccl(\Dom) \\
  &= \sUnion_{\aclass \in \quotcl} \aclass \sunion \sdUnion_{\aclass \in \quotcl} \aclass \cdot \encfunch[\aclass](\Dom)\\
  & (\text{because } \Hist
  = \sUnion_{\aclass \in \quotcl} \aclass) \\ \end{aligned}$ \\ In the
  expression above, $\aclass$
  is a regular language by \Cref{prop:classisregular}, and so is $\encfunch[\aclass](\Dom)$ since $\ainter[\alpha]$ is automatic. Moreover, by
  Proposition~\ref{prop:quotfinite}, the unions are finitely many,
  entailing that $\encfunc(\asDom)$ is a regular language.
\end{proof}

It now remains to prove that $\encfunc$ is a good candidates for the
relations of the structure $\asStruc$.

\begin{lemma}
  \label{lem:allareregular}
Relations $\encfunc(\aseprel)$, $\encfunc(\aspredinter)$, $\encfunc(\fromw)$ and
 $\encfunc(\ofDom)$ are regular.
\end{lemma}
\begin{proof}
%Due to lack of space, we only prove the case of $\encfunc(\epinter)$ and $\encfunc(\aspredinter)$, and we refer to Appendix~\ref{app:lem:allareregular} for the remaining cases.
By the definition of $\asStruc$ from \Cref{sec:fohist-fostruct}, we have:
\begin{itemize}
\item $\encfunc(\epinter) = \encfunc(\epinter*) = \epinter* = (\Hist \times \Hist) \sinter \aeprel[\agent][0] \cdot \edgeformula^*$, where we recall that $\aeprel[\agent][0]$ is the epistemic relation in $\model$, and $\edgeformula^*$ is the set of pointwise concatenation of pairs in the binary relation $\edgeformula$ in $\event$.

Now, because $\Worlds$ is finite, $\aeprel[\agent][0] \subseteq \Worlds \times \Worlds$ is regular. Also, since $\Events$ is finite, $\edgeformula \subseteq \Events \times \Events$ is regular and so is $\edgeformula^*$. Obviously $\Hist \times \Hist$ is regular. Since regular languages are closed under intersection, $\encfunc(\epinter)$ is a regular language.
\item If $\apred(\arity) \in \Pred$,
$\encfunc(\aspredinter) = \encfunc(\aspredinter*) = \encfunc(\sUnion_{\aclass \in \quotcl} \aclass \times \apredinterp[\apred][\aclass]) = \sUnion_{\aclass \in \quotcl} \encfunc(\aclass) \times \encfunc(\apredinterp[\apred][\aclass])$. Now because by definition of $\encfunc$,
$$\encfunc(\apredinterp[\apred][\aclass])=(\prod_{\aintb = 1}^{\arity} \encfunc(\aclass)) \cdot \encfunch[\aclass](\apredinterp[\apred][\aclass]),$$ this set is also equal to $\sUnion_{\aclass \in \quotcl} \encfunc(\aclass) \times ((\prod_{\aintb = 1}^{\arity} \aclass) \cdot \encfunch[\aclass](\apredinterp[\apred][\aclass]))$.
 Since regular languages are closed under cartesian product and
  union, this last expression describes a regular language.
\item  $\encfunc(\fromw) = \encfunc(\Histw) = \Histw = \Hist \sinter \aworld \cdot \Events^*$. which clearly is a regular language.
\item Finally, $\encfunc(\ofDominter) = \encfunc(\ofDominter*) = \sUnion_{\aclass \in \quotcl} \aclass \times \encfunch[\aclass](\Dom)$. As a finite union of regular languages, $\encfunc(\ofDominter)$ is a regular language.
\end{itemize}
\end{proof}

Lemmas~\ref{lem:Domisregular} and~\ref{lem:allareregular} conclude the proof of \Cref{theo:AutStrucH}, and as a consequence, we state the following.
%\end{comment}

As a consequence, we state the following.

\begin{theorem}
\label{theo:EPPzero}
\EPPnmqf is decidable.
\end{theorem}

\begin{example}
\label{ex:EPPextend}
Because an event whose effect is the operation of self-concatenation
-- that we did not consider in \Cref{ex:event} -- would require some
existential quantifier in its post-condition, the query whether a
given language $\wordlang$ can be obtained by union, complement and
concatenation from a finite set of generators, or not,
(see \Cref{ex:EPP}) may not be decidable because
of \Cref{theo:EPPun}.

On the contrary, fixing a bound $k$ to the number of occurences of the
concatenation operation in the expression for $\wordlang$, allows us
to reduce the problem to the decidable epistemic planning
problem \EPPnmqf: we can control by quantifier-free post-conditions
and enough unary predicates $\predlangcompute_0,\ldots,\predlangcompute_k$
in the signature that the concatenation operator is applied at most $k$ times.
%% a bounded number of timeswhere the
%% concatenation is used more than $k$ times, we ensure that there is
%% still a finite number of different interpretations
%% (\ie \Cref{prop:quotfinite} remains true). Therefore the particular
%% case of our problem with a bounded use of concatenation is
%% decidable. \\
\end{example}

Remark that our running example on formal languages does not fully
exploit the epistemic setting, as all events are distinguishable. To
make our formal language problem more exciting on this aspect, one can expand it by
introducing ``errors" in the action model: the agent may not know
whether the applied event is, say, $\eventunioni$ (the union with
languages $L_i$) or some other event corresponding to the union with
an other language $\wordlang[i]'$, say because of
computation errors during the process. We then would be able to decide
whether or not the agent can still decide if the language $\wordlang$
is obtainable through operations on $\wordlang[0]$,...,$\wordlang[n]$.

Still, we believe that our running example on formal languages is an
interesting one for being reminiscent of the many challenging problems
in formal language theory as addressed
in \cite{brzozowski1980open}. In particular, the case of language
concatenation, as considered in \Cref{ex:EPPextend}, is tightly connected with the
\emph{dot-depth hierarchy} of \cite{cohen1971dot}, where many conjectures
remain \cite{pin2017dot}.